\documentclass{article}

\usepackage{amsmath, amssymb, amsthm}
\usepackage{mathptmx,txfonts} 
\usepackage[linesnumbered,ruled,vlined]{algorithm2e}
\usepackage{algorithmic}
\usepackage[margin=1.5in]{geometry}
\usepackage{xcolor}
\usepackage{tikz}
\usepackage{float}
\usepackage{hyperref}
\allowdisplaybreaks

\usepackage{libertine}       % Sets the main text font to Libertine
\usepackage[libertine]{newtxmath} % Sets the math font to match Liberti

\newcommand{\prob}[1]{\mathbb{P}\left(#1 \right)}

\newcommand{\expect}[1]{\mathbb{E}\left[#1\right]}

\newcommand{\rounddown}[1]{\lfloor #1 \rfloor}
\newcommand{\roundup}[1]{\lceil #1 \rceil}
\DeclareMathOperator*{\argmin}{arg\,min}

\newcommand{\tO}[1] {\tilde{O} \left( #1 \right)}

\DeclareMathOperator{\opt}{OPT}
\DeclareMathOperator{\polylog}{polylog}

\newtheorem{theorem}{Theorem}
\newtheorem{lemma}[theorem]{Lemma}

\SetKwInOut{Input}{Input}
\SetKwInOut{Output}{Output}

\title{Constructing Decision Trees from Data Streams \footnote{Authors are listed in alphabetical order.}}
\author{Huy Pham \thanks{Hanoi University of Science and Technology. Email: \texttt{huypham261203@gmail.com}}, 
Hoang Ta \thanks{National University of Singapore, Singapore. Email: \texttt{hoang27@nus.edu.sg}. Supported by the NUS Presidential Young Professorship under Grant No. A-0010021-00-00.   Corresponding Author.}, 
Hoa T. Vu \thanks{San Diego State University, San Diego, CA, USA. Email: \texttt{hvu2@sdsu.edu}. Supported by the National Science Foundation under Grant No. 2342527. Corresponding Author.}
}
\date{}

\begin{document}

\maketitle

\begin{abstract}
In this work, we present data stream algorithms to compute optimal splits for decision tree learning. In particular, given a data stream of observations \(x_i\) and their corresponding labels \(y_i\), without the i.i.d. assumption, the objective is to identify the optimal split \(j\) that partitions the data into two sets, minimizing the mean squared error (for regression) or the misclassification rate and Gini impurity (for classification). We propose several efficient streaming algorithms that require sublinear space and use a small number of passes to solve these problems. These algorithms can also be extended to the MapReduce model. Our results, while not directly comparable, complements the seminal work of Domingos-Hulten (KDD 2000) and Hulten-Spencer-Domingos (KDD 2001).  
\end{abstract}

\section{Introduction}
Decision trees are one of the most popular machine learning models. They serve as the base model for many ensemble methods such as gradient boosting machines \cite{MasonBBF99, Friedman2002}, random forest \cite{Ho95}, AdaBoost \cite{FreundS97}, XGBoost \cite{ChenG16}, LightGBM \cite{KeMFWCMYL17}, and CatBoost \cite{ProkhorenkovaGV18}. These methods are extremely powerful; they yield state-of-the-art results in many machine learning tasks, especially for tabular data \cite{GrinsztajnOV22, Shwartz-ZivA22}.

In the CART \cite{BreimanFOS84} and ID3 \cite{Quinlan93} algorithms, a decision tree is often built in a top-down manner. The tree is grown recursively by splitting the data into two parts at each node depending on different criteria such as Gini impurity, misclassification rate, or entropy for classification and mean squared error for regression.

A key step in building a decision tree is to find optimal splits. Denote $ [N]$ as the set $\{1, 2, \ldots, N\}$. The data consists of observations $ x_1, x_2, \ldots, x_m \in [N] $ and their real-valued labels $ y_1, y_2, \ldots, y_m \in [0, M] $ (for regression) or binary labels $ y_1, y_2, \ldots, y_m \in \{-1, +1\} $ (for classification). Informally, we want to find the optimal split $ j \in [N] $ to partition the data into two sets, $ \{x_i \leq j\} $ and $ \{x_i > j\} $, such that the mean squared error (MSE) for regression or the misclassification rate for classification is minimized.

The set $[N]$ can be thought of as a discretization of the possible values of $ x_i $, which is almost always the case in practice, as computers have finite precision. In fact, the values of observations are often discretized into a finite number of intervals \cite{BreimanFOS84, ChengFIQ88, Quinlan93, HongLK06, DoughertyKS95, FayyadI92}. However, the number of discretized values $ N $ can still be very large, as observed in \cite{FayyadI92}. This is due to the fact that sometimes an attribute is a combination of multiple attributes, or an attribute is a categorical variable with high cardinality. This affects both the memory and time complexity.

In the case of regression, we assume that $ y_i \in [0,M] $ and each $ y_i $ has finite precision that can be stored in one machine word of size $ O(\log N) $ bits.

In many applications, the data is too large to fit into the main memory of a single machine. Often, the data is streamed from disk or over a network while using sublinear space, ideally in a single pass. In some cases, a small number of passes is allowed.

Another computational model of interest is the massively parallel computation (MPC) model. In this model, we have a large number of machines with sublinear memory. In each synchronous round, each machine can communicate with other machines and perform local computations. The MPC model is an abstraction of many distributed computing systems such as MapReduce and Spark \cite{KarloffSV10}.

We first consider a very simple problem of finding the optimal split which is the building block of decision tree learning.  This corresponds to building a depth-1 tree. In the case of classification, this is a decision stump which is often used as a weak learner for AdaBoost \cite{FreundS97}. We can repeat this algorithm to build a deeper tree. We consider the case in which each observation $ x_i $ represents a {\bf single attribute}. Extending this to multiple attributes is straightforward, which we will discuss shortly.

The following formulation of finding the optimal split follows the standard approach from \cite{BreimanFOS84} and \cite{Quinlan93}. For a more recent overview, we refer to chapter 9 of \cite{HastieTF09}. In this work, we consider streaming and massively parallel computation settings.

\paragraph{Compute the optimal split for regression.} Consider a stream of observations $ x_1, x_2, \ldots, x_m \in [N] = \{1,2,\ldots,N\} $ along with their labels $ y_1, y_2, \ldots, y_m \in [0, M] $. Recall that for simplicity (i.e., to avoid introducing another parameter), we assume each $ y_i $ can be stored in one machine word of size $ O(\log N) $ bits.

Let $ [ E ] $ denote the indicator variable for the event $ E $. Furthermore, let $D$ be the number of distinct values of $\{ x_1, x_2, \ldots, x_m \}$. An important primitive to build a regression tree is to find the optimal split $ j \in [N] $ that minimizes the following quantity:
\begin{align*}
    \sum_{i=1}^m [x_i \leq j] (y_i - \mu(j))^2 + \sum_{i=1}^m [x_i > j] (y_i - \gamma(j))^2,
\end{align*}

where $ \mu(j) $ and $ \gamma(j) $ are the mean of the labels of $ x_i $ in the range $[1,2,\ldots,j]$ and $[j+1,j+2,\ldots,N]$, respectively. More formally,
\begin{align*}
	\mu(j) := \frac{\sum_{i=1}^m [x_i \leq j] y_i}{\sum_{i=1}^m [x_i \leq j]}, &  & \gamma(j) := \frac{\sum_{i=1}^m [x_i > j] y_i}{\sum_{i=1}^m [x_i > j]}.
\end{align*}

That is the predicted value for $ x_i $ is $ \mu(j) $ if $ x_i \leq j $ and $ \gamma(j) $ otherwise. This serves as the optimal regression rule for the split at $ j $ based on mean squared error. The loss function based on the mean squared error (MSE) is defined as follows:
\[
	L(j) = \frac{1}{m}\left( \sum_{i=1}^m [x_i \leq j] (y_i - \mu(j))^2 + \sum_{i=1}^m [x_i > j] (y_i - \gamma(j))^2 \right).
\]

Our goal is to find a split $ j $ that minimizes or approximately minimizes the mean squared error $ L(j) $.  We use the notation:
\[
	\opt  = \min_{j \in [N]} L(j) .
\]

In the offline model, we can store the entire data using $\tO{m}$\footnote{We use $\tO{}$ to suppress $\polylog{  N}$ or $\polylog m$ factors whenever convenient. In this paper, the suppressed factors have low orders such as 1 or 2 and are explicit in the proofs.} space and then find the optimal split $ j $ in $ \tO{m} $ time. Our goal is to remedy the memory footprint and to use sublinear space $ o(m) $ and $o(N)$ in the streaming model. 

\begin{figure}
	\begin{tikzpicture}[scale=0.4]
		% Draw axes
		\draw[thick,->] (0,0) -- (11,0) node[anchor=north west] {x};
		\draw[thick,->] (0,0) -- (0,11) node[anchor=south east] {y};
		\foreach \x in {1,...,10} {
				\draw (\x,1pt) -- (\x,-3pt) node[anchor=north] {\x};
				\draw (1pt,\x) -- (-3pt,\x) node[anchor=east] {\x};
			}
		% Draw points
		\fill (1,6) circle (3pt);
		\fill (1,7) circle (3pt);
		\fill (2,8) circle (3pt);
		\fill (3,8) circle (3pt);
		\fill (3,6) circle (3pt);
		\fill (2,7) circle (3pt);
		\fill (4,7) circle (3pt);
		\fill (5,2) circle (3pt);
		\fill (6,2) circle (3pt);
		\fill (6,1) circle (3pt);
		\fill (7,3) circle (3pt);
		\fill (8,2) circle (3pt);
		\fill (9,3) circle (3pt);
		\fill (9,1) circle (3pt);

		\draw[solid] (4,0) -- (4,10);
		\draw[dotted] (0,7) -- (4,7);
		\draw[dotted] (4,2) -- (10,2);
	\end{tikzpicture}
	\hfil
	\begin{tikzpicture}[scale=0.4]
		% Draw axes
		\draw[thick,->] (0,0) -- (11,0) node[anchor=north west] {x};
		\foreach \x in {1,...,10} {
				\draw (\x,1pt) -- (\x,-3pt) node[anchor=north] {\x};
			}
		% Draw points
		\fill (1,1) circle (5pt);
		\draw (1,4) circle [radius=5pt, fill=white, draw=black];
		\fill (2,1) circle (5pt);
		\fill (3,1) circle (5pt);
		\fill (3,1) circle (5pt);
		\draw (2,4) circle [radius=5pt, fill=white, draw=black];
		\fill (4,1) circle (5pt);
		\draw (5,4) circle [radius=5pt, fill=white, draw=black];
		\draw (6,4) circle [radius=5pt, fill=white, draw=black];
		\draw (7,4) circle [radius=5pt, fill=white, draw=black];
		\draw (8,4) circle [radius=5pt, fill=white, draw=black];
		\draw (9,4) circle [radius=5pt, fill=white, draw=black];
		\fill (9,1) circle (5pt);

		\draw[solid] (4,0) -- (4,10);
	\end{tikzpicture}
	\caption{The left figure is an example of regression. The optimal split is $ j = 4 $ which minimizes the mean squared error. The right figure is an example of classification. The optimal split is $ j = 4 $ which minimizes the misclassification rate.}
\end{figure}

Our first results are the following streaming algorithms for regression.

\begin{theorem}[Main Result 1]\label{thm:main-regression}
	For regression, we have the following algorithms:
	\begin{enumerate}
		\item A 1-pass algorithm that uses $ \tO{D} $ space, $ O(1) $ update time, and $ O(D) $ post-processing time that computes the optimal split. $D$ is the number of distinct values of $x_1, x_2, \ldots, x_m$.
		
		\item A 2-pass algorithm that with high probability \footnote{We make a convention that a high probability is at least $ 1- 1/N^{\Omega(1)}$ or $1-1/m^{\Omega(1)}$.} has the followings properties. It uses $ \tO{1/\epsilon} $
		space, $ \tO{1} $ update time, and $ \tO{1/\epsilon} $ post-processing time. Furthermore, it computes a split $ j $ such that $L(j) \leq \opt + \epsilon$. 
		
		\item For any $\beta \in (0,1)$, an $O( 1/\beta)$-pass algorithm  that uses $\tO{1/\epsilon^2 \cdot N^\beta}$ space , $\tO{1/\epsilon^2}$  update time and  post-processing time. It computes a split $ j $ such that $L(j) \leq (1+\epsilon)\opt$.  By setting $\beta = 1/\log N$, this implies an $O(\log N)$-pass algorithm that uses $ \tO{1/\epsilon^2} $ space, $ \tO{1/\epsilon^2} $ update time and post processing time. It computes a split $ j $ such that $L(j) \leq (1+\epsilon)\opt$.
		
	\end{enumerate}
\end{theorem}

\paragraph{Compute the optimal split for classification.} Consider a stream of observations $ x_1, x_2, \ldots, x_m \in [N] $ along with their binary labels $ y_1, y_2, \ldots, y_m \in \{-1, +1\} $. Our goal is to find a split $ j $ that minimizes the  misclassification rate. For any $R \subseteq [N]$ and label $u \in \{-1, +1\}$, we use $f_{u,R}$ to denote the count of the observations with label $u$ in the range $R$. That is
\[
	f_{u,R} := \sum_{i=1}^m [x_i \in R] [y_i = u].	
\]

If we split the observations into two sets $[1,j]$ and $[j+1,N]$, we classify the observations in each set based on the majority label of that set. Then, the  misclassification rate for a given $j$ can be formally defined as follows. 
\begin{align*}
	L(j) & := \frac{1}{m} \cdot \left(\min \{f_{-1, [1,j]}, f_{+1, [1,j]}\} + \min \{f_{-1, (j,N]}, f_{+1, (j,N]}\} \right).
\end{align*}

In a similar manner, we use $\opt$ to denote the smallest possible misclassification rate over all possible $j$. That is $\opt = \min_{j \in [N]} L(j)$. The task is to find a split $j$ that approximately minimizes the misclassification rate, i.e., $L(j) \approx \opt.$

Beside the misclassification rate, another popular loss function is based on the Gini impurity.  Given a (multi)set of observations and labels $S = \{x_i, y_i\}_{i=1}^m$, the Gini impurity is defined as
\[
	\text{Gini}(S) := 1 - \sum_{l \in \{-1,+1\}} \left( \frac{\sum_{i=1}^m [y_i = l]}{m} \right)^2.
\]

Note that when all labels are the same, the Gini impurity is 0. For a given split point $j$, let $L = \{(x_i, y_i) : x_i \leq j\}$ and $R = \{(x_i, y_i) : x_i > j\}$. The loss function based on the Gini impurity is defined as follows:

\[
	L_{\text{Gini}}(j)  = \frac{|L|}{m} \text{Gini}(L) + \frac{|R|}{m} \text{Gini}(R).
\]

Our results for classification are as follows.

\begin{theorem}[Main Result 2] \label{thm:main-classification}
	For classification with numerical observations, we have the following algorithms:
	\begin{enumerate}
		\item A 1-pass algorithm that with high probability satisfies the following properties. It uses $ \tO{1/\epsilon}  $ space, $ O(1) $ update time, and $ \tO {1/\epsilon}$ post-processing time. It outputs a split $ j $ such that $ L(j) \leq \opt + \epsilon $.

		\item For any $\beta \in (0,1)$, an $O( 1/\beta)$-pass algorithm that uses $\tO{1/\epsilon^2 \cdot N^\beta}$ space , $\tO{1/\epsilon^2}$  update time and  post-processing time. It computes a split $ j $ such that $L(j) \leq (1+\epsilon)\opt$. By setting $\beta = 1/\log N$, this implies an $O(\log N)$-pass algorithm that uses $ \tO{1/\epsilon^2} $ space, $ \tO{1/\epsilon^2} $ update time and post processing time. It computes a split $ j $ such  that the loss based on Gini impurity$L(j) \leq (1+\epsilon)\opt$.
		
		\item A 1-pass algorithm that with high probability satisfies the following properties. It uses $ \tO{1/\epsilon^2}  $ space, $ O(1) $ update time, and $ \tO {1/\epsilon^2}$ post-processing time. The algorithm outputs a split $ j $ such that the loss based on Gini impurity $ L_{\text{Gini}}(j) \leq \opt + \epsilon $.
	\end{enumerate}
\end{theorem}

\paragraph{Compute the optimal split for classification with categorical observations.} For categorical observations, we again assume that the observations $ x_1, x_2, \ldots, x_m \in [N]$. However, we can think of $[N]$ as not having a natural ordering (such as locations, colors, etc.).

The goal is to compute a partition of $[N]$ into two disjoint sets $A$ and $B$, denoted by $A \sqcup B = [N]$, such that the misclassification rate is minimized. The misclassification rate is defined as follows:
\begin{align*}
	L(A,B) & := \frac{1}{m} \cdot \left (\min \{f_{-1, A}, f_{+1, A}\} + \min \{f_{-1, B}, f_{+1, B}\} \right).
\end{align*}

Let $\opt$ be the smallest possible misclassification rate over all possible partitions $A$ and $B$. That is $\opt = \min_{A,B:A \sqcup B = [N]} L(A,B)$. We will show the following.

\begin{theorem}[Main Result 3] \label{thm:main-classification-categorical}
	For the classification problem with categorical observations, we have the following results:
	\begin{enumerate}
		\item A 1-pass algorithm that finds a partition $[N] = A \sqcup B$ such that $L(A,B) \leq \opt + \epsilon$ using $\tO{N/\epsilon}$ memory, $O(1)$ update time, and $O(2^N)$ post processing time.
		\item Any constant-pass algorithm that decides if $\opt = 0$ requires $\Omega(N)$ space.
	\end{enumerate}
\end{theorem}

The table in Figure \ref{fig:summary} summarizes our results for the data stream model.

\begin{figure}

	\begin{tabular}{|c|c|c|c|}
		\hline
		Problem & Space & No. of passes &  Output guarantee  \\
		\hline
		R & $\tO{D}$ & 1  & $L(j) = \opt$ \\
		\hline 
		R with $M = O(1)$ & $\tO{1/\epsilon}$ & 2 &   $L(j) \leq \opt + \epsilon$ \\
		\hline
		R \& C.1 & $\tO{1/\epsilon}$ & 1 &  $L(j) \leq \opt + \epsilon$ \\
		\hline
		%R \& C.1 & $\tO{1/\epsilon^2}$ & $O(\log N)$ &  $L(j) \leq (1+\epsilon)\opt$  \\
		%\hline
		R \& C.1 & $\tO{1/\epsilon^2 \cdot N^\beta}$ & $O(1/\beta)$ &  $L(j) \leq (1+\epsilon)\opt$  \\
		\hline
		C.2 & $\tO{1/\epsilon^2}$ & $1$ &  $L_{Gini}(j) \leq \opt + \epsilon$  \\
		\hline
		CC.1 & $\tO{N/\epsilon}$ & 1 &  $L(A,B) \leq \opt + \epsilon$ \\
		\hline
		CC.1 & $\Omega(N)$ & const. &  $L(A,B) \leq \text{const} \cdot \opt $ \\
		\hline
	
	\end{tabular}
	\caption{Result summary. R: Regression, C: Classification, CC: classification with categorical attributes. 1: loss function based on misclassification rate, 2: loss function based on Gini impurity. } 
    \label{fig:summary}
\end{figure}

\paragraph{Extending to multiple attributes.}  In the above, each $x_i$ is the value of {\bf a single attribute}. If we have $d$ attributes, we can simply run our algorithm for each attribute in parallel. 

More formally, if we have $d$ attributes, i.e., $x_i = (x_{i,1}, x_{i,2}, \ldots, x_{i,d})$. In this setting, the optimal split is the best split over all attributes. Let $L_q(j)$ be the misclassification rate (for classification) or the mean squared error (for regression) for the $q$-th attribute at split $j$. Then, the optimal split is given by
\begin{align*}
	\argmin_{q \in [d], j \in [N]} L_q(j).
\end{align*}

In other words, we could find the optimal split for each attribute and then return the best split over all $d$ attributes.  This results in an overhead factor $d$ in the space complexity.

\paragraph{Massively parallel computation model.} In the massively parallel computation (MPC) model, $m$ observations and their labels distributed among $m^{1-\beta}$ machines. Each machine has memory $\tO{m^{\beta}}$.  In each round, machines can communicate with all others, sending and receiving messages, and perform local computations. The complexity of an MPC algorithm is the number of rounds required to compute  the answer. The MPC model is an abstraction of many distributed computing systems such as MapReduce, Hadoop, and Spark \cite{KarloffSV10}.

We can adapt our streaming algorithms to the MPC model. The results are summarized in Figure \ref{fig:summary-mpc}. 

\begin{figure}
	\begin{tabular}{|c|c|c|c|c|}
		\hline
		Problem & Space per machine  & No. of rounds &  Output guarantee  \\
		\hline
		R \& C.1 & $\tO{\sqrt{m}}$  & $O(1)$ & $L(j) \leq \opt + 1/\sqrt{m}$ \\
		\hline
		R \& C.1 & $\tO{1/\epsilon^2 \cdot m^{\beta}}$  & $O(1/\beta)$ & $L(j) \leq (1+\epsilon)\opt$ \\
		\hline 
		C.2 & $\tO{\sqrt{m}}$ & $O(1)$ &  $L_{Gini}(j) \leq \opt + 1/m^{0.25}$  \\
		\hline
	\end{tabular}
	\caption{Result summary. R: Regression, C: Classification, CC: classification with categorical attributes. 1: loss function based on misclassification rate, 2: loss function based on Gini impurity. } 
    \label{fig:summary-mpc}
\end{figure}

\paragraph{Related work and comparison to our results.} Previously, a very popular work by Domingos and Hulten \cite{DomingosH00} \footnote{Their paper was cited more than 3000 times according to Google Scholar.} 
proposed a decision tree learning algorithm for data streams. In their study, they considered the scenario where one has an infinite stream of observations where each observation is independently and identically distributed (i.i.d.). They quantified the number of samples required to find the optimal split and provided a bound on the probability that a sample will follow a different path in the decision tree than it would in a tree built using the entire data set under some assumptions. It is worth noting that the work of Domingos and Hulten \cite{DomingosH00} did not consider the regression problem.

Subsequently, various heuristics, implementations, and applications of streaming decision trees have been developed. Providing a comprehensive list of related work is challenging. Notable examples include Gomes et al. \cite{GomesBRBEPHA17}, who introduced an adaptive random forest algorithm for classification in evolving data streams. Additionally, C++ and Python systems \cite{BifetZFHZQHP17, Montiel2021} have been developed to implement decision trees for data streams, which have seen extensive commercial use. Another popular application of streaming decision trees is mining concept drift \cite{WangFYH03, CanoK22}. Other notable works include \cite{hulten2001mining, jin2003efficient, bifet2009new, rutkowski2014cart, bifet2017extremely, manapragada2018extremely}.

Compared to these works, our approach does not assume an infinite data stream. Additionally, we do not rely on the i.i.d. assumption, making our method more applicable to scenarios where data streams are collected in a time-dependent manner and the data distribution evolves over time. Furthermore, we also consider the regression problem. Our algorithms come with provable guarantees.

\paragraph{Premilinaries and notation.} 
%Suppose we have stream of $ m $ observations $ x_1, x_2, \ldots, x_m \in [N] $. We want to estimate the count of elements in each range $ [a,b] $ for $ a, b \in [N] $.  That is for every $ a, b$, we want to estimate
% \[
% 	f_{[a,b]} := \sum_{i=1}^{m} [a \leq x_i \leq b].
% \]

% We have the following sampling theorem that is a standard application of Chernoff bound whose proof is deferred to Appendix \ref{sec:omitted-proofs}.

% \begin{theorem} \label{thm:sampling}
% 	Suppose we sample each of $ x_1,x_2,\ldots,x_m $ independently with probability $ p = \frac{C \log N}{m^\alpha} $ (for $m^\alpha > C \log N$) for some sufficiently large constant $C$. Let $ k_{[a,b]} $ be the number of sampled elements in the range $ [a,b] $. Then the following holds with probability at least $ 1 - 1/N $ for all $ a, b \in [N] $:
% 	\begin{itemize}
% 		\item If $ f_{[a,b]} \geq  4 m^\alpha $, then $ \frac{k_{[a,b]}}{p} \geq 2m^\alpha$
% 		\item If $ f_{[a,b]} \leq m^\alpha/8 $, then $ \frac{k_{[a,b]}}{p} \leq m^\alpha $.
% 	\end{itemize}
% \end{theorem}

We often use $f_R$ to denote the count of elements in the range $R \subseteq [1,N]$. For classification, we use $f_{-1,R}$ and $f_{+1,R}$ to denote the count of elements with label $-1$ and $+1$ in the range $R$, respectively. 

We also often rely on a common fact that $\log_{1+\epsilon} x = O(1/\epsilon \cdot \log x)$ for all $x \geq 1$ and $\epsilon \in (0,1)$. Recall that we use $[E]$ to denote the indicator variable for the event $E$. 

\paragraph{Handling deletions.} Our algorithms can be extended to handle dynamic stream with deletions. In particular, instead of sampling to estimate the counts of elements or labels in a given range $r \subseteq [N]$, we can use Count-Min sketch  with dyadic decomposition to estimate the counts \cite{CormodeM05}. This will result in an extra $\log N$ factor in the space complexity and the update time.

\paragraph{Paper organization.} In Section \ref{sec:regression}, we present the streaming algorithms for approximating the optimal split for regression. In Section \ref{sec:classification}, we present the streaming algorithms for approximating the optimal split for classification. Finally, Section \ref{sec:parallel} outlines out how to adapt our streaming algorithms to the massively parallel computation setting. All omitted proofs are in Appendix \ref{sec:omitted-proofs}.

\section{Compute the Optimal Split for Regression in Data Streams}\label{sec:regression}

The goal of this section is to design streaming algorithms that compute the optimal split for regression. 

\paragraph{Exact algorithm.} Recall $ D \leq N $ is the number of distinct values of $ x_1, x_2,\ldots, x_m $. We can find the optimal split $ j $ exactly in $ \tO{D} $ space and $ \tO{D} $ time. 

\begin{algorithm}[H]
	\caption{A 1-pass algorithm for finding the optimal split for regression.}
	\label{alg:regression-exact}
	When $ x_i , y_i $ arrive, update the following quantities:
	\begin{align*}
		A_{x_i} & \leftarrow A_{x_i} + 1, &  & B_{x_i} \leftarrow B_{x_i} + y_i, &  & C_{x_i} \leftarrow C_{x_i} + y_i^2.
	\end{align*}

	For each distinct value $ j $ of $ \{ x_1, x_2, \ldots, x_m \} $, compute:
	\begin{align*}
		\mu(j) & = \frac{\sum_{t\leq j} B_j}{\sum_{t\leq j} A_j}, 
		&& \gamma(j) = \frac{\sum_{t > j} B_j}{{\sum_{t > j} A_j}}.
	\end{align*}

	For each distinct value $ j \in \{ x_1, x_2, \ldots, x_m \} $, using dynamic programming, compute
	\begin{align*}
		L(j) & = \frac{1}{m}\left( \sum_{t \leq j} C_t + \mu(j_t)^2 \sum_{t \leq j} A_t - 2\mu(j) \sum_{t \leq j} B_t + \sum_{t > j} C_t + \gamma(j)^2 \sum_{t > j} A_t - 2\gamma(j) \sum_{t > j} B_t \right).
	\end{align*}

	Return $ j = \argmin_{j} L(j) $.

\end{algorithm}

Algorithm \ref{alg:regression-exact} provides the pseudo-code for the first result of Theorem \ref{thm:main-regression}.

\begin{proof}[Proof of Theorem \ref{thm:main-regression} (1)]
	Consider Algorithm \ref{alg:regression-exact}. The space usage is $ \tO{D} $ since we only need to store $ A_t, B_t, C_t $ for each distinct value $ t $ of $ x_1, x_2, \ldots, x_m $. Note that when an observation arrives, we only need to update $ A_t, B_t, C_t $ for the corresponding value of $ t $; this can be done in $ O(1) $ time using hash tables.

    We need to compute $ \sum_{t \leq j} A_t $, $ \sum_{t > j} A_t,  \sum_{t \leq j} B_t $, $ \sum_{t > j} B_t $, $ \sum_{t \leq j} C_t $, and $ \sum_{t > j} C_t $ for all $t$ in $O(D)$ time. This can be easily done using prefix-sum dynamic programming (e.g., $ \sum_{t \leq j} B_t = (\sum_{t < j} B_t) + B_j $ and $ \sum_{t > j} B_t = (\sum_{t > j + 1} B_t) + B_{j + 1} $). This allows us to compute $ \mu(j) $, $ \gamma(j) $, and $L(j)$ for all distinct values $j$ of $\{x_1,x_2,\ldots,x_m\}$.

	% We need to compute $ \mu(j) $ and $ \gamma(j) $ for each of $D$ distinct values $j$ of $\{x_1,x_2,\ldots,x_m\}$. This requires $ O(D) $ time using dynamic programming with the recurrence $ \sum_{t \leq j} B_t = (\sum_{t < j} B_t) + B_j $ and $ \sum_{t > j} B_t = (\sum_{t > j + 1} B_t) + B_{j + 1} $ (similarly, for $ \sum_{t \leq j} A_t $, $ \sum_{t > j} A_t $, $ \sum_{t \leq j} C_t $, and $ \sum_{t > j} C_t $).

	Finally, we need to show that we correctly compute $ L(j) $. We have
	\begin{align*}
		L(j) & = \frac{1}{m}\left( \sum_{i=1}^m [x_i \leq j] (y_i - \mu(j))^2 + \sum_{i=1}^m [x_i > j] (y_i - \gamma(j))^2 \right)                                                                                          \\
		     & = \frac{1}{m}\left( \sum_{i=1}^m \sum_{t \leq j}[x_i = t] (y_i - \mu(j))^2 + \sum_{i=1}^m \sum_{t > j}[x_i = t] (y_i - \gamma(j))^2 \right)                                                                  \\
		     & = \frac{1}{m}\left( \sum_{i=1}^m \sum_{t \leq j}[x_i = t] (y_i^2 + \mu(j)^2 - 2y_i\mu(j)) + \sum_{i=1}^m \sum_{t > j}[x_i = t] (y_i^2 + \gamma(j)^2 - 2y_i\gamma(j)) \right)                                 \\
		     & = \frac{1}{m}\left( \sum_{t \leq j} (C_t + \mu(j)^2 A_t - 2\mu(j) B_t) + \sum_{t > j} (C_t + \gamma(j)^2 A_t - 2\gamma(j) B_t) \right)                                                                           \\
		     & = \frac{1}{m}\left( \sum_{t \leq j} C_t + \mu(j)^2 \sum_{t \leq j} A_t - 2\mu(j) \sum_{t \leq j} B_t + \sum_{t > j} C_t + \gamma(j)^2 \sum_{t \leq j} A_t - 2\gamma(j) \sum_{t \leq j} B_t \right). \qedhere
	\end{align*}
\end{proof}

\paragraph{Additive error approximation for bounded range.} The first algorithm is expensive if $ D $ is large; $ D$ can be as large as $ N $. Here, we shall show that we can approximate the optimal split point $ j $ using $ \tO{1/\epsilon} $ space such that the output split's mean squared error $ L(j) \leq \opt + \epsilon $. This algorithm however requires that the range of the labels $M$ is bounded, i.e., $M = O(1)$ and an additional pass through the data stream.

We will make use of the following technical lemma: the squared distance between the numbers in a set and their mean exhibits a monotonicity property. This lemma will be useful in both additive error approximation and multiplicative error approximation algorithms.

\begin{lemma}\label{lem:monotonicity}
	Let $ S = \{z_1,z_2,\ldots,z_k\} $ be a set of real numbers and $ \mu(S) := \frac{1}{k} \sum_{i=1}^k z_i $ be the mean of the set $ S $. Define
	\[
		g(S) := \sum_{z \in S} (z - \mu(S))^2.
	\]
	For any $ S' \subseteq S $, we have $ g(S') \leq g(S) $.
\end{lemma}

\begin{proof}
	Without loss of generality, suppose $ S' = \{z_1, z_2, \ldots, z_l\} $ for some $ l \leq k $. We have

	\begin{align*}
		g(S) & = \sum_{i=1}^l (z_i - \mu(S))^2 + \sum_{i=l+1}^k (z_i - \mu(S))^2                      \\
		     & = \sum_{i=1}^l (z_i - \mu(S') + \mu(S') - \mu(S))^2 + \sum_{i=l+1}^k (z_i - \mu(S))^2.
	\end{align*}

	Note that \[ (z_i - \mu(S') + \mu(S') - \mu(S))^2 = (z_i - \mu(S'))^2 + (\mu(S') - \mu(S))^2 + 2(z_i - \mu(S'))(\mu(S') - \mu(S)). \]

	We then have
	\begin{align*}
		 & g(S)  = \sum_{i=1}^l (z_i - \mu(S'))^2 + \sum_{i=l+1}^k (z_i - \mu(S))^2  + \sum_{i=1}^l (\mu(S') - \mu(S))^2 + 2(\mu(S') - \mu(S)) \sum_{i=1}^l (z_i - \mu(S'))  \\
		 & = \sum_{i=1}^l (z_i - \mu(S'))^2 + \sum_{i=l+1}^k (z_i - \mu(S))^2 + l (\mu(S') - \mu(S))^2 + 2(\mu(S') - \mu(S)) \underbrace{{\sum_{i=1}^l (z_i - \mu(S'))}}_{0} \\
		 & = \sum_{i=1}^l (z_i - \mu(S'))^2 + \sum_{i=l+1}^k (z_i - \mu(S))^2 + l (\mu(S') - \mu(S))^2                                                                       \\
		 & = g(S') + l (\mu(S') - \mu(S))^2 \geq g(S'). \qedhere
	\end{align*}

\end{proof}

Consider $ j' > j $ where there are not too many observations between $ j $ and $ j' $
. We want to show that  $ L(j) \approx L(j') $ through the following lemma which is one of the key technical parts of this work.

\begin{lemma}\label{lem:split-point}
	For any $ j' > j $, let  $ b = \sum_{i=1}^m [j < x_i \leq j'] $. We have
	\[
		L(j') \leq L(j) + \frac{5b M^2}{4m}.
	\]
\end{lemma}
\begin{proof}
	We have
	\begin{align*}
		m L(j') & = \sum_{i=1}^m [x_i \leq j'] (y_i - \mu(j'))^2 + \sum_{i=1}^m [x_i > j'] (y_i - \gamma(j'))^2                                                     \\
		        & = \sum_{i=1}^m [x_i \leq j] (y_i - \mu(j'))^2 + \sum_{i=1}^m [j < x_i \leq j'] (y_i - \mu(j'))^2 + \sum_{i=1}^m [x_i > j'] (y_i - \gamma(j'))^2   \\
		        & \leq \sum_{i=1}^m [x_i \leq j] (y_i - \mu(j'))^2 + \sum_{i=1}^m [j < x_i \leq j'] (y_i - \mu(j'))^2 + \sum_{i=1}^m [x_i > j] (y_i - \gamma(j))^2.
	\end{align*}

	The last inequality follows from the fact that $ \sum_{i=1}^m [x_i > j'] (y_i - \gamma(j'))^2 \leq \sum_{i=1}^m [x_i > j] (y_i - \gamma(j))^2 $ based on Lemma \ref{lem:monotonicity}. We continue with the above derivation:

	\begin{align*}
		m L(j') & \leq  \sum_{i=1}^m [x_i \leq j] (y_i - \mu(j'))^2 + \sum_{i=1}^m [j < x_i \leq j'] (y_i - \mu(j'))^2 + \sum_{i=1}^m [x_i > j] (y_i - \gamma(j))^2 \\
		        & = \sum_{i=1}^m [x_i \leq j] (y_i - \mu(j) + \mu(j) - \mu(j'))^2 + \sum_{i=1}^m [j < x_i \leq j'] (y_i - \mu(j'))^2                                \\ &  + \sum_{i=1}^m [x_i > j] (y_i - \gamma(j))^2.
	\end{align*}

	Observe that

	\begin{align*}
		  & \sum_{i=1}^m [x_i \leq j] (y_i - \mu(j) + \mu(j) - \mu(j'))^2                                                                                             \\
		= & \sum_{i=1}^m [x_i \leq j] (y_i - \mu(j))^2 + \sum_{i=1}^m [x_i \leq j] (\mu(j) - \mu(j'))^2 + 2\sum_{i=1}^m [x_i \leq j] (y_i - \mu(j))(\mu(j) - \mu(j')).
	\end{align*}

	The last term evaluates to 0 since $ \sum_{i=1}^m [x_i \leq j] (y_i - \mu(j)) = 0 $. Therefore,

	\begin{align*}
		m L(j') & \leq \sum_{i=1}^m [x_i \leq j] (y_i - \mu(j))^2 + \sum_{i=1}^m [x_i \leq j] (\mu(j) - \mu(j'))^2 + \sum_{i=1}^m [j < x_i \leq j'] (y_i - \mu(j'))^2 \\ &  + \sum_{i=1}^m [x_i > j] (y_i - \gamma(j))^2 \\
		        & = m L(j) + \sum_{i=1}^m [x_i \leq j] (\mu(j) - \mu(j'))^2 + \sum_{i=1}^m [j < x_i \leq j'] (y_i - \mu(j'))^2.
	\end{align*}

	Let $ a = \sum_{i=1}^n [x_i \leq j] $ and $ b = \sum_{i=1}^n [j < x_i \leq j'] $. We have

	\begin{align*}
		m L(j') \leq m L(j) + a (\mu(j) - \mu(j'))^2 + b M^2.
	\end{align*}

	Observe that
	\begin{align*}
		(\mu(j) - \mu(j'))^2 & = \left( \frac{\sum_{x_i \leq j} y_i}{a} - \frac{\sum_{x_i\leq j'} y_i}{a+b} \right)^2  = \left( \frac{(a+b)\sum_{x_i \leq j} y_i- a\sum_{x_i\leq j'} y_i}{a(a+b)} \right)^2 \\
		& = \left( \frac{b \sum_{x_i \leq j} y_i - a \sum_{j < x_i\leq j'} y_i}{a(a+b)} \right)^2.
	\end{align*}

	Hence, $ (\mu(j) - \mu(j'))^2 $ is maximized when $ y_i = 0 $ for all $ i $ such that $ x_i \leq j $ and $ y_i = M $ for all $ i $ such that $ j < x_i \leq j' $ or vice versa. In both cases, $ (\mu(j) - \mu(j'))^2 = \left(\frac{ab M}{a(a+b)}\right)^2 = \left(\frac{b M}{a+b}\right)^2$. Therefore,
	\begin{align*}
		m L(j') & \leq m L(j) + a \left( \frac{bM}{a+b} \right)^2 + b M^2.
	\end{align*}

	Note that $ (a+b)^2 \geq 4ab $ and thus $ \frac{b^2}{(a+b)^2} \leq b/(4a)$. Thus,
	\[
		m L(j') \leq m L(j) + a\cdot \frac{b}{4a} \cdot M^2 + b M^2 = m L(j) + \frac{5bM^2}{4} \implies L(j') \leq L(j) + \frac{5bM^2}{4m} . \qedhere
	\]

\end{proof}

The idea is to find a set of $ \approx 1/\epsilon $ candidate split points $ S $ such that the number of observations between any two consecutive split points is at most $ \epsilon m $. Then, we can compute $ L(j) $ for all $ j \in S $ and return the split point $ j $ with the smallest $ L(j) $. This  gives a $ j $ such that $ L(j) \leq \opt + O(\epsilon) $. 
\begin{algorithm}
	\caption{A 2-pass, $\epsilon$ additive approximation the optimal regression split.}
	\label{alg:regression-additive}
	{\bf Pass 1}: Sample each $ x_i $ independently with probability $ p = \frac{C \log N}{m \epsilon} $ for some sufficiently large constant $C$.  \\

	Let $S$ be the set of distinct values in $\{x_i: \text{ $x_i$ is sampled or $x_i+1$ is sampled}\} \cup \{0,N\}$. Let $j_0, j_1, \ldots, j_k$ be the elements of $S$ in increasing order. \\

	% Let $ k_{[a,b]} $ be the number of sampled elements in the range $ [a,b] $ and consider the estimate
	% \[ \tilde{f}_{[a,b]} = \frac{k_{[a,b]}}{p}. \] \\

	% Let $ j_0 \leftarrow 0$ and let $ S = \{j_0\} $. Let $t \leftarrow 0$. \\

	% \While{true}{
	% Let $ j = j_t + 1$.
	% Using binary search over $[N]$, find the smallest $j'$ such that $  \tilde{f}_{[j, j']}  > \epsilon m $. \\
	% If there is no such $j' $, set $j_{t+1} \leftarrow N$, set $S \leftarrow S \cup \{j_{t+1}\} $
	% and stop the loop. \\
	% \If {$ \tilde{f}_{[j, j']} \geq  2 \epsilon m $ \label{cond:1}}{
	% $j_{t + 1} = j'- 1$ and $j_{t + 2} = j'$.\\
	% $ S \leftarrow S \cup \{j_{t + 1}\} \cup \{j_{t + 2}\} $. \\
	% $t \leftarrow t + 2$.
	% }
	% \Else{
	% 	$j_{t + 1} = j'$. \\
	% 	$ S \leftarrow S \cup \{j_{t+1}\} $. \\
	% 	$t \leftarrow t + 1$.
	% }
	% }

	{\bf Pass 2}: For each $t$ such that $ j_t, j_{t+1} \in S $, compute:
	\begin{align*}
		A_t = \sum_{i = 1}^m [j_t < x_i \leq  j_{t+1}], &  & B_t =\sum_{i = 1}^m [j_t <  x_i \leq j_{t+1}] y_i, 
		& & C_t =\sum_{i = 1}^m [j_t <  x_i \leq j_{t+1}] y_i^2.
	\end{align*} \\
	For each $t$ such that $ j_t, j_{t+1} \in S $, using dynamic programming, compute
	\begin{align*}
		L(j_t)  & = \frac{1}{m} \left( \sum_{l \leq t - 1} C_l  +   \mu(j_t)^2  \sum_{l \leq t - 1} A_l - 2 \mu(j_t)  \sum_{l \leq t - 1}  B_l + \sum_{l  \geq t} C_l + \gamma(j_k)^2 \sum_{l  \geq t} A_t - 2 \gamma(j_t) \sum_{l  \geq t} B_l \right).
  	\end{align*}

	Return $ j_t $ with the smallest $ L(j_t) $. \label{step:final-step}
\end{algorithm}

We now prove the second algorithm in the first main result. 

\begin{proof}[Proof of Theorem \ref{thm:main-regression} (2)]

	Consider Algorithm \ref{alg:regression-additive}. In the first pass, the algorithm samples $O(1/\epsilon \cdot \log N)$ observations with high probability by Chernoff bound. The update time is clearly $O(1)$.
	
	Consider any interval $[a,b]$ where $a,b \in [N]$. If the numbers of observations $x_i \in [a,b]$ is at least $\epsilon m/2$, for some sufficiently large constant $C$, the probability that we do not sample
	any $x_i \in [a,b]$ is at most
	\[
		\left(1 - \frac{C \log N}{m \epsilon}\right)^{\epsilon m/2} \leq e^{-C \log N/2} < 1/N^{4}.
	\]

	Therefore, by appealing to the union bound over at most ${N \choose 2}  \leq N^2$ choices of $a, b$ we have that with probability at least $1-1/N^2$, for all $a,b\in [N]$ if $|\{x_i : a \leq x_i \leq b\}| > \epsilon m/2$, then we sample at least one $x_i \in [a,b]$. We can condition on this event since it happens with high probability.

	We next establish the approximation guarantee. Note that if we sample some $x_i = j$, then we add both $j$ and $j - 1$ to the candidate set.  Suppose the optimal split $j^\star \in S$ then we are done since $L(j^\star)$ will be computed in the second pass. Otherwise, there exists $t$ such that $j_t < j^\star < j_{t+1}$. This implies that $j_t < j_{t+1} - 1$; hence, it must be the case that $j_{t+1}$ was added to $S$ because we sample some $x_i = j_{t+2}$ and we did not sample any $x_i = j_{t+1}$. Therefore, the algorithm did not sample any observation in $(j_{t}, j_{t+1}] \supset (j^\star, j_{t+1}]$. With $a=j_\star + 1$ and $b = j_{t+1}$, this implies 
	\[
		|\{ x_i : j_\star < x_i \leq j_{t+1} \}| \leq \epsilon m/2.
	\]

	From Lemma \ref{lem:split-point}, we have
	\[
		L(j_{t+1}) \leq L(j^\star) + \frac{5 \cdot \epsilon m/2}{4m} \cdot M^2 < \opt +  \epsilon M^2  = \opt + O(\epsilon).
	\]

	The last equality uses the assumption that $ M = O(1) $. Finally, we need to show  that Step \ref{step:final-step} is correct and can be done efficiently. Note that there is a trivial way to do this if we allow $ O(1/\epsilon \cdot \log N) $ update time and an extra pass. We however can optimize the update time as follows. In the second pass, we compute
	\begin{align*}
		A_t = \sum_{i = 1}^m [j_t < x_i \leq  j_{t+1}], &  & B_t =\sum_{i = 1}^m [j_t <  x_i \leq j_{t+1}] y_i, &  & C_t =\sum_{i = 1}^m [j_t <  x_i \leq j_{t+1}] y_i^2.
	\end{align*}
	for all $ j_t, j_{t+1} \in S $. 
	Note that when $x_i, y_i$ arrives, we can update the corresponding $ A_t, B_t, C_t $ in $\tO{1} $ time since $ x_i $ is in exactly one interval $ (j_t, j_{t+1}] $ which we can find in $ O(\log \log N + \log 1/\epsilon))$ time using binary search since $S = O(\epsilon^{-1} \log N)$ with high probability.

	It is easy to see that from $ A_t $, $ B_t $, and $ C_t $, we can compute $ \mu(j_t) $ and $ \gamma(j_t) $ for all $ j_t \in S $. That is
	\begin{align*}
		\mu(j_t) = \frac{\sum_{l \leq t-1} B_l}{\sum_{l \leq t-1} A_l}, &  & \gamma(j_t) = \frac{\sum_{l \geq t} B_l}{\sum_{l \geq t} A_l}.
	\end{align*}

	We observe the following:
	\begin{align*}
		  & m L(j_t) =   \sum_{i = 1}^m [ x_i \leq j_{t}] (y_i - \mu(j_t))^2  + \sum_{i = 1}^m [ x_i > j_{t}] (y_i - \gamma(j_k))^2 .                                                                                                                                                                         \\
		= & \sum_{i = 1}^m  [ x_i \leq j_{t}] (y_i^2  + \mu(j_t)^2 - 2 y_i \mu(j_t))  + \sum_{i = 1}^m [ x_i > j_{t}] (y_i^2 + \gamma(j_t)^2 - 2 y_i \gamma(j_t) )                                                                                                                                            \\
		= & \sum_{i = 1}^m \sum_{l \leq t - 1} [j_l < x_i \leq j_{l+1}] (y_i^2  + \mu(j_t)^2 - 2 y_i \mu(j_t))  + \sum_{i = 1}^m \sum_{l \geq t} [j_l < x_i \leq j_{l+1}] (y_i^2 + \gamma(j_t)^2 - 2 y_i \gamma(j_t) )                                                                                        \\
		= & \left( \sum_{l \leq t - 1} C_l  \right)+   \mu(j_t)^2 \left( \sum_{l \leq t - 1} A_l \right) - 2 \mu(j_t)  \left( \sum_{l \leq t - 1}  B_l \right)  + \left(\sum_{l  \geq t} C_l \right) + \gamma(j_k)^2 \left( \sum_{l  \geq t} A_t \right) - 2 \gamma(j_t) \left( \sum_{l  \geq t} B_l \right).
	\end{align*}

	Therefore, we can compute $ L(j_k) $ for all $ j_k \in S $ after the second pass based on the above quantities in $ O(1/\epsilon \cdot \log N) $ time using dynamic programming. This completes the proof.
\end{proof}

\paragraph{Multiplicative errror approximation.} We now show how to obtain a factor $1+\epsilon$ approximation algorithm for the regression problem. A multiplicative approximation is better if the mean squared error is small; however, this comes at a cost of $O(\log N)$ passes and an extra $1/\epsilon$ factor in the memory use.

The main idea is to guess the squared errors on the left and right sides of the optimal split point and use binary search to find an approximate solution using Lemma \ref{lem:monotonicity}.

At a high level, the algorithm guesses the left and right squared errors of the optimal split point up to a factor of $(1+\epsilon)$. It then uses binary search to find a split point satisfying that guess. For a particular guess, the algorithm either finds a feasible split point or correctly determines that the guess is wrong. 

\begin{algorithm}[H]
	\caption{An $O(\log N)$-pass, factor $1+\epsilon$ approximation for the optimal regression split.}
	\label{alg:regression-multiplicative}

	\ForEach{$z_{left}, z_{right} \in \{0,(1+\epsilon),(1+\epsilon)^2,\ldots, \roundup{\log_{1+\epsilon} (mM^2)} \}$} {
		$j_l \leftarrow 1, j_r \leftarrow N$.  \\
		$S \leftarrow \emptyset$. \tcp{$S$ is a hash map}
		\While {$j_l \leq j_r$}{
			$j \leftarrow \lfloor (j_l + j_r)/2 \rfloor$.\\
			Make a pass over the stream and compute $\mu(j)$ and $\gamma(j)$.\\
			Make another pass over the stream and compute
			\begin{align*}
				\texttt{Error}_{left}(j)  =  \sum_{i = 1}^m [x_i \leq j] (y_i - \mu(j))^2, &&
				\texttt{Error}_{right}(j)  = \sum_{i = 1}^m [x_i > j] (y_i - \gamma(j))^2.
			\end{align*}
			\If {$\texttt{Error}_{left}(j) \leq z_{left}$ and $\texttt{Error}_{right}(j) \leq z_{right}$} {
				\tcp{$j$ is a split with squared error at most $z_{left} + z_{right}$.}
				\If {$j$ is not in $S$ or $S[j] > z_{left} + z_{right}$}{
					$S[j] \leftarrow z_{left} + z_{right}$.
				}
			} \ElseIf {$\texttt{Error}_{left}(j) > z_{left}$ and $\texttt{Error}_{right}(j) > z_{right}$}{
				$z_{left},z_{right}$ are a wrong guess. Skip this guess.
			} \ElseIf {$\texttt{Error}_{left}(j) > z_{left}$ and $\texttt{Error}_{right}(j) \leq z_{right}$}{
				$j_r \leftarrow j - 1$.
			} \ElseIf {$\texttt{Error}_{left}(j) \leq z_{left}$ and $\texttt{Error}_{right}(j) > z_{right}$}{
				$j_l \leftarrow j + 1$.
			}
		}
	}
	Return $j \in S$ such that $S[j]$ is minimized. 
\end{algorithm}

\begin{lemma}\label{lem:regresion-special-case}
	There exists an $O(\log N)$-pass algorithm that uses $ \tO{1/\epsilon^2} $ space, $ \tO{1/\epsilon^2} $ update time and post processing time. It computes a split $ j $ such that $L(j) \leq (1+\epsilon)\opt$ 
\end{lemma}
\begin{proof}
	Consider Algorithm \ref{alg:regression-multiplicative} running on a guess $z_{left}, z_{right}$.
	The algorithm makes $O(\log N)$ passes over the stream since it simulates a binary search on $[1,N]$. Furthermore, the update time is $O(1)$ and the memory is $\tO{1}$.

	There are at most $O(1/\epsilon^2 \cdot \log^2 (mM^2))$ guesses in total. Thus, we run $\tO{1/\epsilon^2}$ instances of the algorithm in parallel. As a result, the total memory and update time are both $\tO{1/\epsilon^2}$. 

	We now show that the algorithm outputs a $j$ such that $L(j) \leq (1+\epsilon) \opt$. Let $j^\star$ be the optimal split and define
	\begin{align*}
		z^\star_{left}  := \sum_{i = 1}^m [x_i \leq j^\star] (y_i - \mu(j^\star))^2, &&
		z^\star_{right} := \sum_{i = 1}^m [x_i > j^\star] (y_i - \gamma(j^\star))^2.
	\end{align*}
	There must be a desired guess $z_{left}, z_{right}$ such that $z^\star_{left} \leq z_{left} \leq (1+\epsilon) z^\star_{left}$ and $z^\star_{right} \leq z_{right} \leq (1+\epsilon) z^\star_{right}$. 

	Consider any iteration of the binary search for a fixed $z_{left}$ and $z_{right}$. Let $j$ be the current search point. Note That
	\[
		L(j) = \frac{1}{m} \left(\texttt{Error}_{left}(j) + \texttt{Error}_{right}(j) \right).	
	\]
	We consider the following cases:
	\begin{itemize}
		\item Case 1: Suppose $\texttt{Error}_{left}(j) \leq z_{left}$ and $\texttt{Error}_{right}(j) \leq z_{right}$.  We say $j$ is a feasible split for this guess. Then, $m L(j) \leq  z_{left} + z_{right}$. For the desired guess, this implies that $ m L(j) \leq (1+\epsilon) (z^\star_{left} + z^\star_{right})$ which implies $L(j) \leq (1+\epsilon) \opt$.
		
		\item Case 2: Suppose $\texttt{Error}_{left}(j) > z_{left}$ and $\texttt{Error}_{right}(j) > z_{right}$. This guess must be wrong. Because of Lemma \ref{lem:monotonicity}, for any split point to the right of $j$, the left error will larger than $z_{left}$. Similarly, for any split point to the left of $j$, the right error larger than $z_{right}$. Therefore, we can skip this guess since there is no feasible split.

		\item Case 3.1: Suppose $\texttt{Error}_{left}(j) > z_{left}$ and $\texttt{Error}_{right}(j) \leq z_{right}$. We know that from Lemma \ref{lem:monotonicity}, for any split point to the right of $j$, the left error will be larger than $z_{left}$. Therefore, we can skip all split points to the right of $j$ and recurse on the left side.
		
		\item Case 3.2: Suppose $\texttt{Error}_{left}(j) \leq z_{left}$ and $\texttt{Error}_{right}(j) > z_{right}$. This case is similar to the previous case. We know that from Lemma \ref{lem:monotonicity}, for any split point to the left of $j$, the right error will be larger than $z_{right}$. Therefore, we can skip all split points to the left of $j$ and recurse on the right side.
	\end{itemize}
	Finally, we can return the feasible split $j$ with the smallest guaranteed squared error. This completes the proof.
\end{proof}

We obtain the final part of the first main result by simulating multiple iterations of the binary search in  fewer passes at the cost of an increased memory.

\begin{proof}[Proof of Theorem \ref{thm:main-regression} (3)]

	This result is also based on Algorithm \ref{alg:regression-multiplicative}. Recall that we allow $O(1/\epsilon^2 \cdot N^{\beta})$ space for some fixed $0 < \beta < 1$. Fix a guess $z_{left}, z_{right}$.
	In the first iteration of the binary search, we check if the split $j = \rounddown{(1+N)/2}$ is feasible. If not, in the second iteration, there are 2 possible values of $j$ that the binary search can check for feasibility, depending on the outcome of the first iteration. In general, after $r$ iterations, the total number of possible splits $j$ that the binary search checks for feasibility is $1+2+4+\ldots+2^r = 2^{r+1} - 1 $. These values partition $[N]$ into $2^r$ intervals of length at most $ \approx N/2^r$.
	
	Note that the feasibility of a split $j$ can be checked based on its left and right squared errors. Instead of running $r$ iterations of binary search in $2r$ passes, we can compute the left and right squared errors of each of these $2^{r+1} - 1$ splits in 2 passes using $\tO{ 2^r}$ space. The update time is $O(1)$ by using the same trick as in Algorithm \ref{alg:regression-additive}. This allows us to simulate $r$ iterations of the binary search in $2$ passes.

	The search interval after $r$ binary search iterations has size at most $\approx N/2^r$. We repeat this process on the remaining search interval. After we repeat this process $t$ times, the size of the search interval is at most $\approx N/2^{rt}$ this implies $t = O( \frac{\log N}{r})$. If we set $2^r = N^\beta$, then $t = O(1/\beta)$. Therefore, we use at most $O(1/\beta)$ passes. Finally, recall that there are $\tO{1/\epsilon^2}$ guesses which results in a total memory of $\tO{1/\epsilon^2 \cdot N^\beta}$.
	
	Since we need to run $\tO{1/\epsilon^2}$ instances of the algorithm in parallel, the total update time is $\tO{1/\epsilon^2}$. We note that the post-processing time is also $\tO{1/\epsilon^2}$ since we need to simulate $\tO{1/\epsilon^2}$ binary searches corresponding to $ \tO{1/\epsilon^2}$ guesses. 
\end{proof}

\section{Compute the Optimal Split for Classification in Data Streams} \label{sec:classification}

\subsection{Numerical observations}

\paragraph{Misclassification rate.} Suppose we have an input stream $\{(x_1, y_1), (x_2, y_2),\ldots, (x_m, y_m)\}$. Each $x_i$ is an observation in the range $[N]$ with a label $y_i \in \{-1, +1\}$. Our goal is to find a split point $j$ that minimizes the  misclassification rate.

Recall that the misclassification rate for a given $j$ is defined as 
\begin{align*}
	L(j) &  := \frac{1}{m } \cdot \left(\min \{f_{-1, [1,j]}, f_{+1, [1,j]}\} + \min \{f_{-1, (j,N]}, f_{+1, (j,N]}\} \right).
\end{align*}

We use $\opt = \min_i L(j) $ to denote the smallest possible misclassification rate over all possible $j$. Our goal is to find an $j$ that approximately minimizes the misclassification rate.

\begin{algorithm}
	\caption{A 1-pass, $\epsilon$ additive error approximation of the optimal classification  split.}
	\label{alg:classifcation-additive}
	Sample each $ x_i, y_i $ independently with probability $p= \frac{C \log N}{\epsilon m} $ for some sufficiently large constant $C$.  \\

	\ForEach{$f_{l, r} \in \{f_{-1, [1,j]}, f_{-1, (j,N]}, f_{+1, [1,j]}, f_{+1, (j,N]}\}_{j \in [N]}$}
	{
	Let $ k_{l,r} $ be the number of sampled elements with corresponding label $ l $  in range $ r $. \\
	Define the estimate
	\[ \tilde{f}_{l, r} = \frac{k_{l,r}}{p}. \]
	}

	Return $ j $ among the values of sampled $x_i$ that minimizes $ \tilde{L}(j) $ where
	\[
		\tilde{L}(j) := \frac{1}{m} \cdot \left( \min \{ \tilde{f}_{-1, [1,j]}, \tilde{f}_{+1, [1,j]} \} +\min \{ \tilde{f}_{-1, (j,N]}, \tilde{f}_{+1, (j,N]} \} \right) .
	\]
\end{algorithm}

\paragraph{Additive error approximation for the loss function based on misclassification rate.} We now prove the first part of Theorem \ref{thm:main-classification}. We make use of Hoeffding's inequality for the bounded random variables below.

\begin{theorem}\label{thm:hoeffding}
	Let $ X_1, X_2, \ldots, X_m $ be independent random variables such that $ a_i \leq X_i \leq b_i $ for all $ i \in [m] $. Let $\mu = \expect{\sum_{i=1}^m X_i}$.
	Then for any $ t > 0 $, we have
	\[
		\prob{\left|\sum_{i=1}^m X_i - \mu \right| \geq t} \leq e^{-2t^2/\sum_{i=1}^m (b_i - a_i)^2}.
	\]
\end{theorem}

The proof of the first part of Theorem \ref{thm:main-classification} is a basic application of Theorem \ref{thm:hoeffding} and the union bound.

\begin{proof}[Proof of Theorem \ref{thm:main-classification} (1)]
	Consider Algorithm \ref{alg:classifcation-additive}. For each $ f_{l,r} $, from Theorem \ref{thm:hoeffding},
	\begin{align*}
		\prob{\left| k_{l,r} -  p f_{l,r} \right| \geq \epsilon/4 \cdot p m} & \leq 2 e^{-2 \epsilon^2/16 \cdot p^2 m^2/ f_{l,r}} \leq 2 e^{- \epsilon^2/8 \cdot p^2 m}.
	\end{align*}

	By choosing $ p = \frac{C \log N}{\epsilon m} $ for some sufficiently large constant $C$, we have 
	\[
		\prob{\left| k_{l,r} -  p f_{l,r} \right| \geq \epsilon/8\cdot p m} \leq 1/N^4.
	\]
	Taking the union bound over $2$ choices of $l$ ($+1$ and $-1$) and $N$ choices of $r$ (corresponding to $N$ split points), we have that all the estimates $ \tilde{f}_{l,r} $ satisfy
	\[
		f_{l,r} - \epsilon m/8 \leq \tilde{f}_{l,r} \leq f_{l,r} + \epsilon m/8.
	\]
	with probability $ 1-1/N $. For the optimal split $j^\star $ and corresponnding classification choices $ u^\star, v^\star \in \{-1, +1\} $, we obtain
	\begin{align*}
		f_{u^\star, [1,j^\star]} + f_{v^\star, (j^\star, N]} -\epsilon m/8 \leq \tilde{f}_{u^\star, [1,j^\star]} + \tilde{f}_{v^\star, (j^\star, N]} & \leq f_{u^\star, [1,j^\star]} + f_{v^\star, (j^\star, N]} + \epsilon  m/8.
	\end{align*}

	Therefore,
	\[
		\tilde{L}(j^\star) = \frac{1}{m} \cdot \left( \min \{ \tilde{f}_{+1, [1,j^\star]}, \tilde{f}_{-1, [1,j^\star]} \} + \min \{ \tilde{f}_{+1, (j^\star,N]}, \tilde{f}_{-1, (j^\star,N]} \} \right) \leq \opt + \epsilon/4.
	\]

	Hence, there must exist at least one $ j $ satisfying $ \tilde{L}(j) \leq \opt +\epsilon/4 $, specifically the optimal $ j^\star $. Conversely, assume for some $u,v \in \{-1,+1\}$, we have
	\begin{align*}
		\tilde{L}(j) = \frac{\tilde{f}_{u, [1,j]} + \tilde{f}_{v, (j,N]}}{m} \leq \opt + \epsilon/4.
	\end{align*}

	Since $ \tilde{f}_{u, [1,j]} \geq f_{u, [1,ij]} - \epsilon m/8 $ and $ \tilde{f}_{v, (j,N]} \geq f_{v, (j,N]} - \epsilon m/8 $, we have
	\begin{align*}
		L(j) = \frac{{f}_{u, [1,j]} + {f}_{v, (j, N]} - \epsilon m/4}{m} \leq \opt + \epsilon/4 \implies L(j) < \opt + \epsilon.
	\end{align*}

	Clearly, the update time is $ O(1) $. The algorithm samples $ O( mp ) = O \left( \frac{\log N} {\epsilon} \right)$ observations and labels with high probability by Chernoff bound. Since storing each $ x_i, y_i $ requires $ O(\log N) $ bits, the space complexity is $ O \left(\frac{\log^2 N }{\epsilon} \right)$.

	The post-processing time is $ \tO  {\frac{\log N }{\epsilon}}$. We can sort the distinct values $j_1,j_2,\ldots,j_d$ that appear in the samples in $\tO{ \frac{\log N }{\epsilon} }$ time; without loss of generality, assume $j_1 < j_2 < \ldots < j_d$. Note that $d = O\left( \frac{\log N }{\epsilon} \right)$ with high probability. Furthermore, we only need to consider candidate split $j \in \{j_1,j_2,\ldots,j_d\}$ since if $j \in (j_t, t_{t+1})$ then $\tilde{f}_{l, [1,j]} = \tilde{f}_{l, [1,j_t]}$ and $\tilde{f}_{l, (j,N]} = \tilde{f}_{l, (j_{t},N]}$ for $l = \pm 1$.
	
	For $u \in \{-1,+1\}$, we compute the number of sampled observations $(x_i, y_i)$ where $x_i = j_t$ and $y_i = u$ denoted by $k_{u,[j_t,j_t]}$. For $u \in \{+1,-1\}$ and each values $j_t$, we have the recurrence
    \begin{align*}
        k_{u,[1,j_t]} & = k_{u,[1,j_{t-1}]} + k_{u, [j_t, j_t]} \\
        k_{u,(j_t,N]} & = k_{u,(j_{t+1},N]} + k_{u, [j_{t+1}, j_{t+1}]}.
    \end{align*} 
    This allows us to compute $ \tilde{L}(j_t) $ for all $j_t$ in $\tO{ \frac{\log N }{\epsilon}} $ time using dynamic programming.
\end{proof}

\paragraph{Additive error approximation for the loss function based on Gini impurity.} We can also use the Gini impurity instead of the misclassification rate. Given a set of observations $O= \{x_i, y_i\}_{i=1}^m$, the Gini impurity is
\[
	\text{Gini}(O) = 1 - \sum_{l \in \{-1,+1\}} \left( \frac{\sum_{i=1}^m [y_i = l]}{m} \right)^2.
\] 
Recall that for a split $j$, $L = \{(x_i, y_i) : x_i \leq j\}$, and $R = \{(x_i, y_i) : x_i > j\}$, the loss function based on the Gini impurity of the split is
\[
	L_{Gini}(j) = \frac{|L|}{m} \text{Gini}(L) + \frac{|R|}{m} \text{Gini}(R).
\]

Devising a sampling strategy to estimate the Gini impurity of a split point $j$ is non-trivial. We first expand the above equation. Note that $|L| = f_{[1,j]}$ and $|R| = f_{(j,N]}$. We have

\begin{align*}
	L_{Gini}(j) & = \frac{|L|}{m} \text{Gini}(L) + \frac{|R|}{m} \text{Gini}(R) \\
				& = 1 - \frac{ f_{-1, [1,j]}^2 }{m f_{[1,j]}} - \frac{ f_{+1, [1,j]}^2 }{m f_{[1,j]}} - \frac{ f_{-1, (j,N]}^2 }{m f_{(j,N]}} - \frac{ f_{+1, (j,N]}^2 }{m f_{(j,N]}}. 
\end{align*}

We now show how to estimate $\frac{ f_{+1, [1,j]}^2 }{m f_{[1,j]}}$ up to an additive error of $\epsilon$. The other terms can be estimated similarly. 

Let $C$ and $K$ be some sufficiently large constants such that $C/2 > K$. We sample $\frac{C \log n} {\epsilon}$ observations. We observe that if an interval has a reasonably large number of observations, then we sample enough observations in that interval.

\begin{lemma}\label{lem:enough-samples}
	Fix a constant $K$. Suppose we sample at least $\frac{C \log N}{\epsilon}$ observations uniformly at random for some sufficiently large constant $C$. With probability at least $1-1/N^2$, for all $j \in [N]$, if $f_{[1,j]} \geq \epsilon m$ (or $f_{(j,N]} \geq \epsilon m$), then we sample $K/\epsilon \cdot \log N$ observations $x_i$ with $x_i \leq j$ (or $x_i > j$ respectively).
\end{lemma}
\begin{proof}
	Assume that $f_{[1,j]} \geq \epsilon m$. In expectation, we sample at least $\epsilon m \cdot \frac{C \log N}{\epsilon m} = C \cdot \log N$ data points $x_i \leq j$. 
	Thus, by Chernoff bound (see Appendix), the probability that we sample fewer than $C /(2\epsilon) \cdot \log N > K/\epsilon\cdot \log N$ data points with $x_i \leq j$ is at most $1/N^3$. The claim follows from a union bound over $N$ values of $j$.
\end{proof}

We independently make two sets of samples of size $C/\epsilon^2 \cdot \log N$. Let the first set be $S_1$ and let $S_2 = \{ (x_i, y_i) \text{ in the second set of samples and}: x_i \leq j\}$. We use Algorithm \ref{alg:term-gini-estimate} to estimate $\frac{f_{+1, [1,j]}^2}{m f_{[1,j]}}$ for all $j \in [N]$ based on these samples.

\begin{algorithm}\label{alg:term-gini-estimate}
	\caption{$\epsilon$ additive error approximation to $\frac{f_{+1, [1,j]}^2}{m f_{[1,j]}}$ for all $j \in [N]$ in one pass.}
	\label{alg:classifcation-gini}
	Let $C,K$ be sufficiently large constants where $C > 2K$. \\

	$S_1, S_2 \leftarrow C/\epsilon^2 \cdot \log N$ observations sampled uniformly at random from the stream using Reservoir sampling. \\ 
	
	\ForEach{$j \in [N]$}{
		$S_{2,j} = \{ (x_i, y_i) \in S_2: x_i \leq j\}$. \\
		\If{there are fewer than $T=K/\epsilon^2 \cdot \log N$ samples in $S_{2,j}$}{
			Set the estimate $est \left(\frac{f_{+1, [1,j]}^2}{m f_{[1,j]}} \right) = 0$. 
		}
		Otherwise, let $T = K/\epsilon^2 \cdot \log N$. \\ 
		Let $a_i = \text{[$i$th sample in $S_{2,j}$ has label +1]}$ for $i = 1,2,\ldots, T$.  \\
		Let $b_i = \text{[$i$th sample in $S_1$ has label +1 and is in the interval $[1, j]$]}$ for $i = 1,2,\ldots, T$. \\
		Let $c_i = a_i \cdot b_i$. \\
		Let $est \left(\frac{f_{+1, [1,j]}^2}{m f_{[1,j]}} \right) = \frac{1}{T} \sum_{i=1}^{T} c_i$.

	}
\end{algorithm}

\begin{lemma}
	In Algorithm \ref{alg:term-gini-estimate}, with probability at least $1-1/N^2$, for all $j \in [N]$, we have
	\begin{align*}
		est \left( \frac{f_{+1, [1,j]}^2}{m f_{[1,j]}} \right) = \frac{f_{+1, [1,j]}^2}{m f_{[1,j]}} \pm \epsilon.
	\end{align*}
\end{lemma}

\begin{proof}
	Let $T = K/\epsilon^2 \cdot \log N$. Note that if $f_{[1,j]} \geq \epsilon m$, then with high probability, we sample at least $T$ observations with labels $+1$ in the interval $[1,j]$. We have
	\begin{align*}
		\expect{\sum_{i=1}^T c_i} & = \sum_{i=1}^T  \expect{b_i} \expect{a_i} = \frac{f_{+1, [1,j]}}{m} \cdot \frac{f_{+1, [1,j]}}{f_{[1,j]}} \cdot T = \frac{f_{+1, [1,j]}^2}{m f_{[1,j]}} \cdot T.
		% \expect{\sum_{i=1}^T c_i} & = 
	\end{align*}
	Therefore, from Hoeffding's inequality, we have
	\begin{align*}
		\prob{\left| \sum_{i=1}^T c_i - \frac{f_{+1, [1,j]}^2}{m f_{[1,j]}} \cdot T \right| \geq \epsilon T} & \leq 2 e^{-2 \epsilon^2 T^2/T} \leq 2 e^{-2 \epsilon^2 T} \leq 1/N^3.
	\end{align*}
	This implies that with probability at least $1-1/N^3$, we have:
	\begin{align*}
		est \left( \frac{f_{+1, [1,j]}^2}{m f_{[1,j]}} \right) = \frac{f_{+1, [1,j]}^2}{m f_{[1,j]}} \pm \epsilon.
	\end{align*}
	On the other hand, if $f_{[1,j]} < \epsilon m$, then 
	\begin{align*}
		\frac{f_{+1, [1,j]}^2}{m f_{[1,j]}} = \frac{f_{+1, [1,j]}}{m} \times \frac{f_{+1, [1,j]}}{f_{[1,j]}} \leq \epsilon \times 1 = \epsilon.
	\end{align*}
	In this case, the claim is trivially true since our estimate is 0. Appealing to a union bound over $N$ values of $j$ completes the proof.
\end{proof}

We put these together to complete the proof of Theorem \ref{thm:main-regression} (3).

\begin{proof}[Proof of Theorem \ref{thm:main-regression} (3)]
	Since for all $j \in [N]$, we can approximate each term in the loss function based on the Gini impurity up to an additive error of $\epsilon$, we are able to approximate the optimal split point based on the Gini impurity up to an additive error of $O(\epsilon)$. Reparameterize $\epsilon \leftarrow \epsilon/4$ gives us the desired error.
	
	The update time is $O(1)$ and the space complexity is $O(1/\epsilon^2)$. In the post-processing time, we only need to consider the $O(1/\epsilon \cdot \log N)$ distinct values of $j$ that appear in the samples. This gives us a $O(1/\epsilon \cdot \log N)$ post-processing time and completes the proof of the first part of Theorem \ref{thm:main-regression} (3).
\end{proof}

\paragraph{Multiplicative error approximation for the loss function based on misclassification rate.} We now prove the second part of Theorem \ref{thm:main-classification} that gives a factor $1+\epsilon$ approximation algorithm for finding the optimal split based on misclassification rate instead of an additive error. The approach is fairly similar to the regression case.

\begin{algorithm}
	\caption{A $O(\log N)$-pass, $1+\epsilon$ factor approximation for the optimal classification split.}
	\label{alg:classification-multiplicative}
	$S \leftarrow \emptyset$. \tcp{$S$ is a hash map} 
	\ForEach {$z_{left}, z_{right} \in \{ 0, 1+\epsilon, (1+\epsilon)^2,(1+\epsilon)^3,\ldots, \roundup{\log_{1+\epsilon} m} \}$}
	{
		$j_l = 1, j_r = N$. \\
		\While{$i_l \leq i_r$}
		{
            $j = \lfloor (j_l + j_r)/2 \rfloor$. \\
            Make a pass over the stream to compute $f_{-1, [1,j]}, f_{+1, [1,j]}, f_{-1, (j,N]}, f_{+1, (j,N]}$. \\

			$f_u = \min \{ f_{-1, [1,j]}, f_{+1, [1,j]} \}$ and $f_v = \min \{ f_{-1, (j,N]}, f_{+1, (j,N]} \}$.\\
			
			\If{$f_{u} \leq z_{left}$ and $ f_{v} \leq z_{right}$}
			{
				\tcp{$j$ is a split with misclassification rate at most $z_{left} + z_{right}$.}
				\If{$j$ is not in $S$ or $S[j] > z_{left} + z_{right}$}
				{
					$S[j] = z_{left} + z_{right}$.
				}
			}
			\ElseIf{$f_{u} > z_{left}$ and $ f_{v} > z_{right}$}
			{
				$z_{left},z_{right}$ are a wrong guess. Skip this guess.
			}
			\ElseIf{$f_{u} > z_{left}$ and $ f_{v} \leq z_{right}$}
			{
				$j_r = j - 1$.
			}
			\ElseIf{$f_{u} \leq z_{left}$ and $ f_{v} > z_{right}$}
			{
				$j_l = j + 1$.
			}
		}
	}
	Return $j \in S$ such that $S[j]$ is minimized. 
\end{algorithm}

\begin{lemma}
	There exists an $O(\log N)$-pass algorithm that uses $ \tO{1/\epsilon^2} $ space, $ \tO{1/\epsilon^2} $ update time and post processing time. It computes a split $ j $ such that $L(j) \leq (1+\epsilon)\opt$.
\end{lemma}
\begin{proof}[Proof of Theorem \ref{thm:main-classification} (2)]
	Consider Algorithm \ref{alg:classification-multiplicative}. We will show that the algorithm returns a split point $j$ such that $L(j) \leq (1 + \epsilon)\opt$.

	Let $j^\star$ be the optimal split. Let $z^\star_{left}$ and $z^\star_{right}$ be the number of misclassified observations in $[1,j^\star]$ and in $(j^\star,N]$ respectively. 
	
	One of the guesses $z_{left}, z_{right}$ must satisfy $z_{left} \in [z^\star_{left},(1+\epsilon)z^\star_{left}]$ and $z_{right} \in [z^\star_{right}, (1+\epsilon)z^\star_{right}]$ respectively. Fix a guess $z_{left}, z_{right}$. Consider any iteration of the binary search with current search point $j$. 
	
	Let $f_{u, [1,j]} = \min \{ f_{-1, [1,j]}, f_{+1, [1,j]} \}$ and $f_{v, (j,N]} = \min \{ f_{-1, (j,N]}, f_{+1, (j,N]} \}$. We consider the following cases:
	\begin{itemize}
		\item Case 1: If $f_{u} \leq z_{left}$ and $ f_{v} \leq z_{right}$ then we know that $ j $ is a feasible split point that misclassifies at most $z_{left} + z_{right}$ observations. For the desired guess, this implies $L(j) \leq (1+\epsilon) \opt$. We say $j$ is a feasible split for this guess.

		\item Case 2: If $f_{u}  > z_{left}$ and $f_{v} > z_{right}$, then we know that this guess is wrong. This is because any split to the right of $j$ will misclassify more than $z_{left}$ observations on the left and any split to the left of $j$ will misclassify more than $z_{right}$ observations on the right. Therefore, we can skip this guess. We say this guess is infeasible.
		
		\item Case 3.1: If $f_{u} \leq z_{left}$ and $f_{v} > z_{right}$, then we know any split to left of $j$ will misclassify more than $z_{right}$ observations on the right. Therefore, we can skip all splits to the left of $j$ and recurse on the right side.

		\item Case 3.2: If $f_{u} > z_{left}$ and $f_{v} \leq z_{right}$, then we know any split to right of $j$ will misclassify more than $z_{left}$ observations on the left. Therefore, we can skip all splits to right  of $j$ and recurse on the left side.
	\end{itemize}

	For each guess, the update time is $O(1)$, the space complexity is $\tO{1}$, and the number of passes is $O(\log N)$. Since there are $\tO{1/\epsilon^2}$ guesses, the total memory and update time are both $\tO{1/\epsilon^2}$. \qedhere

\end{proof}

Finally, we can simulate multiple iterations of the binary search to yield Theorem \ref{thm:main-classification} (2). The proof is completely analogous to the proof of Theorem \ref{thm:main-regression} (3) and is omitted.

\subsection{Categorical observations}

For categorical observations, we can also use sampling to find the optimal split point up to an additive error $\epsilon$.

\begin{proof}[Proof of Theorem \ref{thm:main-classification-categorical} (1)]

Suppose we sample each $ x_i, y_i $ independently with probability $p= \frac{C N}{m\epsilon} $ for some sufficiently large constant $C$. Similar to the proof of Theorem \ref{thm:main-classification}, we can show that for each $l \in \{-1,+1\}, A \subseteq [N]$, the estimate $ \tilde{f}_{l, A} $ satisfies:
\[
	f_{l, A} - \epsilon m/8 \leq \tilde{f}_{l, A} \leq f_{l, A} + \epsilon m/8.
\]

For optimal $[N] = A^\star \sqcup B^\star$, we have
\[
	\tilde{L}(A^\star) = \frac{1}{m} \cdot \left(\min \{ \tilde{f}_{+1, A^\star}, \tilde{f}_{-1, A^\star} \} + \min \{ \tilde{f}_{+1, B^\star}, \tilde{f}_{-1, B^\star} \} \right) \leq \opt + \epsilon/4.
\]

Conversely, if $ \tilde{L}(A) + \tilde{L}(B) \leq \opt + \epsilon/4$, then $ L(A) + L(B) \leq \opt + \epsilon$. This happens with probability at least $1-1/e^{-10N}$ for some sufficiently large constant $C$.
The only difference here is that we need to take a union bound over $2^N$ choices of $A$ to get the probability bound. 
\end{proof}

\paragraph{Lower bound.} For $N \in \mathbb{N}$, we consider the $\mathrm{DISJ}(N)$ problem in communication complexity. In this problem, Alice has input $a \in \{0,1\}^{N}$ and Bob has input $b \in \{0,1\}^N$. Their goal is to determine whether there is a index $i \in [N]$ such that $a_i = b_i = 1$ or not. Specifically, if there exists an index $i$ such that $a_i = b_i = 1$, the output is YES; otherwise, the output is NO.

In the (randomized) communication setting, Alice and Bob need to communicate to determine whether there exists an index $i \in [N]$ such that $a_{i} = b_{i} = 1$ with an error probability of at most $\frac{1}{3}$.

The communication complexity of $\mathrm{DISJ}(N)$ is the smallest number of bits communicated (in the worst case over $(a, b) \in \{0,1\}^N \times \{0,1\}^N $) by any protocol that determines $\mathrm{DISJ}(N)$.
\begin{theorem}[\cite{ChakrabartiKS03,Bar-YossefJKS04}]
\label{thm:disjoiness}
    The communication complexity of the $\mathrm{DISJ}(N)$ problem is $\Omega(N)$, even for randomized protocols.
\end{theorem}

Consider a stream of observations $ x_1, x_2, \ldots, x_m \in [N] $ and their labels $ y_1, y_2, \ldots, y_m \in \{-1,+1\}$. Using Theorem~\ref{thm:disjoiness}, we will show that any streaming algorithm for approximating the optimal misclassification rate, denoted by $\mathrm{OPT}$, with categorical observations requires at least $\Omega(N)$ space. To accomplish this, we provide a reduction from the $\mathrm{DISJ}(N)$ problem to problem of deciding whether $\opt >0$ in the case of categorical observations. Suppose $\mathcal{S}$ is a streaming algorithm that decides whether $\opt = 0$. The following reduction proves the second part of Theorem \ref{thm:main-classification-categorical}.

\begin{proof}[Proof of Theorem \ref{thm:main-classification-categorical} (2)]
	For each $a_i = 1$, Alice generates an observation label pair $(i, +1)$ and feeds it to the streaming algorithm $\mathcal{S}$. She then sends the memory state of $\mathcal{S}$ to Bob. For each $b_i = 1$, Bob generates an observation label pair $(i, -1)$ and feeds it to $\mathcal{S}$. Bob then sends the memory state of $\mathcal{S}$ back to Alice and continue the next pass of the streaming algorithm.

	If there is no index $i$ such that $a_i = b_i = 1$, then $\opt = 0$. We can choose $A = \{ i : a_i = 1\}$ and $B = \{ i : b_i = 1\}$. Then, $L(A,B) = 0$ since all observations in $A$ are labeled $+1$ and all observations in $B$ are labeled $-1$. Thus, $\opt = 0$.

	If there is an index $i$ such that $a_i = b_i = 1$, then there are two observations $(i, +1)$ and $(i, -1)$ generated by Alice and Bob respectively. Thus, at least one of them is misclassified. Therefore, $\opt > 0$.

	Therefore, any constant-pass streaming algorithm that distinguishes between the case of $\opt = 0$ and $\opt > 0$ requires at least $\Omega(N)$ space.
\end{proof}

\section{Massively Parallel Computation Algorithms} \label{sec:parallel}

In the massively parallel computation (MPC) model, we have $m$ samples distributed among $m^{1-\beta}$ machines. Each machine has memory $\tO{m^{\beta}}$. In each round, each machine can send and receive messages from all other machines. The complexity of an MPC algorithm is the number of rounds required to compute  the answer.

\paragraph{Regression.} We first adapt the algorithms in Theorem \ref{thm:main-regression} to the MPC model. We use the following utility procedure in our MPC algorithms whose proof is deferred to Appendix \ref{sec:omitted-proofs}.

\begin{lemma} \label{lem:mpc-utility}
	Suppose there are $m^{1-\beta}$ machines each of which has memory $\tO{m^\beta}$. Consider $k = O(m^\beta)$ candidate splits $j_1 < j_2 < \ldots < j_k$. We can compute $L(j_t)$, $\texttt{Error}_{left}(j_t)$, and $\texttt{Error}_{right}(j_t)$ for each $t \in [k]$ in $O(1)$ rounds.
\end{lemma}

We state our main result for regression in the MPC model as a theorem below.

\begin{theorem}\label{mpc:regression}
	For regresion, in the MPC model, we have the following algorithms:
	\begin{enumerate}
		\item If the number of machines is $\sqrt{m}$ and each machine has memory $\tO{\sqrt{m}}$ and the label range $M=O(1)$, then there exists an MPC algorithm that computes a split  $j$ such that $L(j) \leq \opt + O(\frac{1}{\sqrt{m}})$ in $O(1)$ rounds.
		\item If the number of machines is $m^{1-\beta}$ and each machine has memory $\tO{1/\epsilon^2 \cdot m^\beta}$, then there exists an MPC algorithm that computes a split $j$ such that $L(j) \leq (1+\epsilon)\opt $ in $O(1/\beta)$ rounds.
	\end{enumerate}

\end{theorem}

\begin{proof}[Proof of Theorem \ref{mpc:regression}  (1)]
	Each machine independently samples each $x_i, y_i$ it holds with probability $p = \frac{C \log N}{\sqrt{m}}$ for some sufficiently large constant $C$. With high probability, every machine samples at most $\tO{1}$ observation and label pairs. These machines will send their samples to a central machine. Note that the central machine will receive $\tO{\sqrt{m}}$ samples with high probability; thus, this does not violate the memory constraint.

	The central machine will find the candidate splits $S=\{j_1, j_2, \ldots, j_k\}$ where $k = O(\sqrt{m})$ as in the second algorithm in Theorem \ref{thm:main-regression}. It will then send these candidate splits to all machines. Note that each machine receives a message of size $\tO{\sqrt{m}}$. Since we have $O(\sqrt{m})$ candidate splits and $\sqrt{m}$ machines, we can compute $L(j_t)$ for each $j_t$ in $O(1)$ rounds using Lemma \ref{lem:mpc-utility}.

	This in turns allows the central machine to find the split $j$ that minimizes $L(j)$ among $j_1, j_2, \ldots, j_k$ which guarantees that $L(j) \leq \opt + O(\frac{1}{\sqrt{m}})$ as argued in the proof of Theorem \ref{thm:main-regression}.
\end{proof}

We now prove the second part of Theorem \ref{mpc:regression}.

\begin{proof}[Proof of Theorem \ref{mpc:regression}  (2)]
	We can simulate the third algorithm in Theorem \ref{thm:main-regression} in the MPC model as follows. Recall that we try $O(1/\epsilon^2 \cdot \log^2 N)$ guesses $z_{left}, z_{right}$ for the left and right errors of the optimal split.
	Fix a guess, we run a binary search to find a feasible split $j$ or correctly determine that the guess is infeasible. We can simulate $r =\beta \log N$ iterations of the binary search by checking the feasibility of each of $2^{r+1}-1 = O(N^\beta)$ splits. Let these splits be $j_1 < j_2 < \ldots < j_{2^{r+1}-1}$. 
	From Lemma \ref{lem:mpc-utility}, we can compute the left and right errors $\texttt{Error}_{left}(j_t)$ and $\texttt{Error}_{right}(j_t)$ for each $j_t$ in $O(1)$ rounds which in turn allows us to determine the feasibility of each split with respect to this guess.
	
	Therefore, we are able to simulate $\beta \log N$ binary search iterations in $O(1)$ rounds. Hence, the total number of rounds is $O(1/\beta)$. Note that since each machine has $\tO{1/\epsilon^2 \cdot m^\beta}$ memory, we can run this algorithm for all guesses simulatenously.
\end{proof}

\paragraph{Classification.} The first algorithm in Theorem \ref{thm:main-classification} can easily be adapted to the MPC model with $\sqrt{m}$ machines. Each machine samples each $x_i, y_i$ that it holds with probability $p = \frac{C \log N}{\sqrt{m}}$ for some sufficiently large constant $C$ and sends its samples to a central machine. The central machine will then compute the optimal split $j$ that minimizes the misclassification rate based on the sampled observations and labels as in the first algorithm in Theorem \ref{thm:main-classification}. This requires 1 MPC-round. 

The third algorithm in Theorem \ref{thm:main-classification} can be simulated in 1 MPC-round as well given that we have $\sqrt{m}$ machines each with memory $\tO{\sqrt{m}}$. Each machine can independently sample each $x_i, y_i$ that it holds with probability $p = \frac{C \log N}{\sqrt{m}}$ and send its samples to a central machine. With high probability, the central machine will receive $\tO{\sqrt{m}}$ samples. The central machine then computes the best split $j$ that minimizes the loss function based on the Gini impurity based on the samples as in the third algorithm in Theorem \ref{thm:main-classification}; this is effectively equivalent to setting $\epsilon = 1/m^{1/4}$.

We can also simulate the third algorithm in Theorem \ref{thm:main-classification} in the MPC model.  The proofs are completely analogous to the regression case and are omitted.

\begin{theorem}
	For classification, in the MPC model, we have the following algorithms:
	\begin{enumerate}
		\item If the number of machines is $\sqrt{m}$ and each machine has memory $\tO{\sqrt{m}}$, then there exists an MPC algorithm that computes a split point $j$ such that $L(j) \leq \opt + \frac{1}{\sqrt{m}}$ in $1$ round.
		
		\item If the number of machines is $\sqrt{m}$ and each machine has memory $\tO{\sqrt{m}}$, then there exists an MPC algorithm that computes a split point $j$ such that $L(j) \leq \opt + \frac{1}{{m^{1/4}}}$ in $1$ round.
		
		\item If the number of machines is $m^{1-\beta}$ and each machine has memory $\tO{1/\epsilon^2 \cdot m^\beta}$, then there exists an MPC algorithm that computes a split point $j$ such that $L(j) \leq (1+\epsilon)\opt $ in $O(1/\beta)$ rounds.
	\end{enumerate}
	
\end{theorem}

\appendix

\section{Omitted Proofs}\label{sec:omitted-proofs}

% We will use the following Chernoff bound for the binomial distribution.
% \begin{theorem}\label{thm:chernoff}
% 	Let $ X \sim \text{Binomial}(Q,p) $. Then for any $ 0 < \epsilon < 1 $, we have
% 	\[
% 		\prob{X  \leq  (1 - \epsilon)pQ } \leq   e^{- \epsilon^2 pQ/2}.
% 	\]
% 	Furthermore, for any $ \epsilon > 0 $, we have
% 	\[
% 		\prob{X \geq (1 + \epsilon)pQ} \leq \left( \frac{e^{ \epsilon }}{(1+\epsilon)^{(1+\epsilon)}} \right)^{pQ}.
% 	\]
% \end{theorem}

% The above lets us prove the following.

%We have the following sampling theorem that is a standard application of Chernoff bound whose proof is deferred to Appendix \ref{sec:omitted-proofs}.

\begin{theorem} [Chernoff bound] \label{thm:sampling}
	Suppose we sample each of $ x_1,x_2,\ldots,x_m $ independently with probability $ p = \frac{C \log N}{m^\alpha} $ (for $m^\alpha > C \log N$) for some sufficiently large constant $C$. Let $ k_{[a,b]} $ be the number of sampled elements in the range $ [a,b] $. Then the following holds with probability at least $ 1 - 1/N $ for all $ a, b \in [N] $:
	\begin{itemize}
		\item If $ f_{[a,b]} \geq  4 m^\alpha $, then $ \frac{k_{[a,b]}}{p} \geq 2m^\alpha$
		\item If $ f_{[a,b]} \leq m^\alpha/8 $, then $ \frac{k_{[a,b]}}{p} \leq m^\alpha $.
	\end{itemize}
\end{theorem}

\begin{proof}%[Proof of Theorem \ref{thm:sampling}   ]

	Let $ k_{[a,b]} $ be the number of sampled elements in the range $ [a,b] $. Then
	\[ k_{[a,b]} \sim \text{Binomial}(f_{[a,b]}, p) .\]

	Suppose $ f_{[a,b]} \geq m^\alpha/2 $; then, $ p f_{[a,b]} \geq  C/2 \cdot \log N $. Hence, for some sufficiently large constant $C $,
	\begin{align*}
		\prob{\frac{k_{[a,b]}}{p} < \frac{m^\alpha}{4}} & =  \prob{k_{[a,b]} < \frac{m^\alpha p}{4}} \leq \prob{k_{[a,b]} < \frac{f_{[a,b]}p}{2}} 
		\leq e^{- \frac{1/4}{2} f_{[a,b]} p} \leq e^{ C/32 \cdot  \log N} < \frac{1}{N^4}.
	\end{align*}

	In the last step, we apply the first inequality in the Chernoff bound with $ \epsilon = 1/2 $.

	On the other hand, suppose $ f_{[a,b]} \leq m^\alpha/8 $ (which implies $\frac{m^\alpha}{f_{[a,b]}} \geq 8$); then,

	\begin{align*}
		\prob{\frac{k_{[a,b]}}{p} > m^\alpha} & =  \prob{k_{[a,b]} > m^\alpha p}  = \prob{k_{[a,b]} > \left(1 + \frac{m^\alpha}{f_{[a,b]}} - 1\right) f_{[a,b]}p}                                                                             \\
		                                      & \leq \left( \frac{e^{m^\alpha/f_{[a,b]}- 1 }}{(m^\alpha/f_{[a,b]})^{(m^\alpha/f_{[a,b]})}} \right)^{f_{[a,b]}p} < \left( \frac{e}{m^\alpha/f_{[a,b]}} \right)^{f_{[a,b]}p m^\alpha/f_{[a,b]}} \\
		                                      & = \left( \frac{e}{m^\alpha/f_{[a,b]}} \right)^{p m^\alpha}< 2^{C \cdot \log N} < \frac{1}{N^4}.
	\end{align*}

	In the second inequality, we apply the second inequality of Chernoff bound with $\epsilon = m^\alpha/f_{[a,b]} - 1 $.

	By taking a union bound over all $N \choose 2$ choices of $ a, b \in [N] $, we have the desired result.
\end{proof}

\begin{proof} [Proof of Lemma \ref{lem:mpc-utility}]
	In MPC, we can sort $(x_i, y_i)$ based on $x_i$ in $O(1)$ rounds and store them on the machines $\mathcal{M}_1, \mathcal{M}_2, \ldots, \mathcal{M}_{m^\beta}$ in order \cite{ghaffari2019massively,GoodrichSZ11}. 

	Let $j_0 = 0$. Using the same argument as in the proof of Theorem \ref{thm:main-regression} (2), in order to compute $L(j_1), L(j_2),\ldots,L(j_k)$, it suffices to compute the following quantities:
	\begin{align*}
		A_t = \sum_{i=1}^m [j_t < x_i \leq j_{t+1}], &  & B_t = \sum_{i=1}^m [j_t < x_i \leq j_{t+1}] y_i, &  & C_t = \sum_{i=1}^m [j_t < x_i \leq j_{t+1}] y_i^2
	\end{align*}
	for each $j_t, j_{t+1}$. Each machine $\mathcal{M}_q$ computes
	\begin{align*}
		A_{t,q} = \sum_{x_i \in \mathcal{M}_q} [j_t < x_i \leq j_{t+1}], &  & B_{t,q} = \sum_{x_i \in \mathcal{M}_q} [j_t < x_i \leq j_{t+1}] y_i, &  & C_{t,q} = \sum_{x_i \in \mathcal{M}_q} [j_t < x_i \leq j_{t+1}] y_i^2
	\end{align*}

	and sends these quantities to the central machine. The central machine then computes $A_t, B_t, C_t$ by aggregating $A_{t,q}, B_{t,q}, C_{t,q}$ over $q$. We still need to argue that the central machine receives a combined message of size $\tO{m^\beta}$.

	We say that machine $\mathcal{M}_q$ is responsible for $j_t$ if it holds some $x_i$ such that $j_{t-1} < x_i \leq j_t$. It is possible that two or more machines may be responsible for the same $j_t$. Suppose machine $\mathcal{M}_q$ is responsible for $j_l < j_{l+1} < \ldots < j_{l'}$. Then, $j_{l+1}, j_{l+2}, \ldots, j_{l'-1}$ are only responsible for by $\mathcal{M}_q$ because the observations are stored in the machines in sorted order. We say $j_l$ and $j_{l'}$ are the boundary splits responsible by $\mathcal{M}_q$. Each machine has at most two boundary splits. Therefore, 
		
	\[
		\left| \left\{ t,q: \text{$j_t$ is a boundary split  responsible for by $\mathcal{M})_q$} \right\} \right| \leq 2 m^\beta.
	\]

	Because each non-boundary split is responsible for by at most one machine and there are $O(m^\beta)$ splits being considered, we have
	\[
		\left| \left\{ t,q: \text{$j_t$ is a non-boundary split  responsible for by $\mathcal{M}_q$} \right\} \right| \leq O(m^\beta).
	\]

	Hence, the total message size received by the central machine is $\tO{m^\beta}$. The central machine then computes  $A_t, B_t, C_t$ and computes $L(j_t)$ for each $j_t$. From the calculation in the proof of Theorem \ref{thm:main-regression} (2), we can also compute $\texttt{Error}_{left}(j_t)$ and $\texttt{Error}_{right}(j_t)$ since
	\begin{align*}
		\texttt{Error}_{left}(j_t) & =   \sum_{i = 1}^m [ x_i \leq j_{t}] (y_i - \mu(j_t))^2 = \left( \sum_{l \leq t - 1} C_l  \right)+   \mu(j_t)^2 \left( \sum_{l \leq t - 1} A_l \right) - 2 \mu(j_t)  \left( \sum_{l \leq t - 1}  B_l \right), \\
		\texttt{Error}_{right}(j_t) &=   \sum_{i = 1}^m [ x_i > j_{t}] (y_i - \mu(j_t))^2 = \left( \sum_{l \geq t} C_l  \right)+   \mu(j_t)^2 \left( \sum_{l \geq t} A_l \right) - 2 \mu(j_t)  \left( \sum_{l \geq t}  B_l \right). \qedhere
	\end{align*}
\end{proof}

\bibliography{references}
\bibliographystyle{plain}
\end{document}